%% file: journal.tex
\documentclass{llncs}

\input{makros_journal}

\begin{document}

\frontmatter

\title{On the diameter of hyperbolic random graphs}

\author{Tobias Friedrich$^{1}$ \and Anton Krohmer$^{1}$}


\institute{
    \small
    \begin{tabular}{c}
        $^1$ Hasso Plattner Institute, Potsdam, Germany
    \end{tabular}
}

\maketitle

\pagestyle{headings}
\pagenumbering{arabic}

\begin{abstract}
Large real-world networks are typically scale-free.
Recent research has shown that such graphs are described best in a geometric space.
More precisely, the internet can be mapped to a hyperbolic space such that geometric greedy
routing performs close to optimal (\citeauthor{boguna2010sustaining}. Nature Communications, 1:62, 2010).
This observation pushed the interest in hyperbolic networks as
a natural model for scale-free networks.
Hyperbolic random graphs follow a power-law degree distribution
with controllable exponent~$\beta$
and show high clustering
(\citeauthor{gugelmann2012random}. ICALP, pp.\ 573--585, 2012).

\medskip
\qquad
For understanding the structure of the resulting graphs and for analyzing the behavior of network algorithms,
the next question is bounding the size of the diameter.
The only known explicit bound is
$\Oh((\log n)^{32/((3-\beta)(5-\beta)) + 1})$
(\citeauthor{KiwiMitsche15}. ANALCO, pp.\ 26--39, 2015).
We present two much simpler proofs for an improved upper bound of $\Oh((\log n)^{2/(3-\beta)})$
and a lower bound of $\Omega(\log n)$.
\end{abstract}

\section{Introduction}

Large real-world networks are almost always sparse and non-regular. Their degree distribution
typically follows a \emph{power law}, which is synonymously used for being \emph{scale-free}.
Since the 1960's, large networks have been studied in detail and hundreds of models
were suggested. In the past few years, a new line of research emerged, which showed
that scale-free networks can be modeled more realistically when incorporating \emph{geometry}.

\textbf{Euclidean random graphs.}
It is not new to study graphs in a geometric space.
In fact, graphs with \emph{Euclidean geometry} have been studied intensively for more than a decade.
The standard Euclidean model are random geometric graphs which result
from placing $n$~nodes independently and uniformly at random on an Euclidean
space, 
and creating edges between pairs of nodes 
if and only if their distance is at most some fixed threshold~$r$.
These graphs have been studied in relation to subjects
such as cluster analysis, statistical physics, hypothesis testing,
and wireless sensor networks~\cite{penrose2003random}.
The resulting graphs are more or less regular and hence
do not show a scale-free behavior with power-law degree distribution
as observed in large real-world graphs.

\begin{table}[t]
\newcommand{\tikzmark}[1]{\tikz[overlay,remember picture] \node (#1) {};}
\begin{tabular}{@{\ }lc}
\toprule
\textbf{Random Graph Model}  & \textbf{Diameter} \\ \midrule
Sparse Erd\H{o}s-R\'{e}nyi 
\cite{bollobas1998random} 
& $\Theta(\log n)$~\cite{riordan2010diameter}  \\
$d$-dim.\ Euclidean
\cite{penrose2003random}
& $\Theta(n^{1/d})$~\cite{diamEucRGG} \\
Watts-Strogatz
\cite{WattsStrogatz98} 
& $\Theta(\log n)$~\cite{BollobasC88}  \\
Kleinberg
\cite{Kleinberg00}
 & $\Theta(\log n)$~\cite{MartelN04}  \\
\midrule
Chung-Lu
\cite{chung2002average} 
&  $\Theta(\log n)$~\cite{chung2002average} \tikzmark{topbrace} \\
Pref. Attachment
\cite{barabasi1999emergence}
& $\Theta(\log \log n)$~\cite{diamPA}  \\
Hyperbolic
\cite{krioukov2010hyperbolic}
 & $\Oh((\log n)^{\frac{32}{(3-\beta)(5-\beta)}+1})$~\cite{KiwiMitsche15} \tikzmark{bottombrace}\tikzmark{right} 
\\
\bottomrule
\end{tabular}
\hspace*{2cm}
\begin{tikzpicture}[overlay, remember picture]
  \draw [decoration={brace,amplitude=0.5em},decorate,thick,black]
    let \p1=(topbrace), \p2=(bottombrace) in
    ({max(\x1,\x2)}, {\y1+0.8em}) -- node[right=0.6em] {power-law graphs} ({max(\x1,\x2)}, {\y2});
\end{tikzpicture}
\caption[test]{
Known diameter bounds for various random graphs.
In all cases the diameter depends on the choice of the model parameters.
Here we consider a constant average degree.
For scale-free networks, we also assume a power law exponent $2 < \beta < 3$.\footnotemark
\vspace{-2em}
}
\label{tab:diameters}
\end{table}

\textbf{Hyperbolic random graphs.}
For modeling scale-free graphs, it is natural to apply a non-Euclidean geometry with negative curvature.
\citet{krioukov2010hyperbolic} introduced a new graph model based on \emph{hyperbolic geometry}. Similar to euclidean random graphs,
nodes are uniformly distributed in a hyperbolic space and two nodes are connected
if their hyperbolic distance is small. 
The resulting graphs have many properties observed in large real-world networks.
This was impressively demonstrated by \citet{boguna2010sustaining}:
They computed a maximum likelihood fit of the internet graph in the hyperbolic space
and showed that greedy routing in this hyperbolic space finds nearly optimal shortest paths in the internet graph.
The quality of this embedding is an indication that hyperbolic geometry naturally appears in large scale-free graphs.

\textbf{Known properties.}
\footnotetext{Note that the table therefore refers to a non-standard Preferential Attachment version with adjustable power law exponent $2<\beta<3$ (normally, $\beta = 3$).} 
A number of properties of hyperbolic random graphs have been studied.
\citet{gugelmann2012random} compute exact asymptotic expressions for the expected number of vertices of degree $k$ and prove a constant lower bound for the clustering coefficient.
They confirm that the clustering is non-vanishing and that the degree sequence follows a power-law distribution with controllable exponent $\beta$. For $2 < \beta < 3$, the hyperbolic random graph has a giant component of size $\Omega(n)$ \cite{Bode:2013aa,bfmconnected}, similar to other scale-free networks like Chung-Lu \cite{chung2002average}. Other studied properties include the clique number \cite{friedrich2015cliques}, bootstrap percolation~\cite{candellero2014bootstrap}; as well as algorithms for efficient generation of hyperbolic random graphs \cite{von2015fast} and efficient embedding of real networks in the hyperbolic plane \cite{6705650}.

\textbf{Diameter.}
The diameter, the length of the longest shortest path, is a fundamental property of a network.
It also sets a worst-case lower bound on the number of steps required for all communication processes
on the graph. In contrast to the average distance, it is determined by a single---atypical---long path.
Due to this sensitivity to small changes, it is notoriously hard to analyze.
Even subtle changes to the graph model can make an exponential difference in the diameter, as can be seen when comparing
Chung-Lu (CL) random graphs~\cite{chung2002average} and
Preferential Attachment (PA) graphs~\cite{barabasi1999emergence} 
in the considered range of the  power law exponent $2 < \beta < 3$:
On the one hand, we can embed a CL graph in the PA graph and they behave effectively
the same~\cite{UltraFastRumor:unpub};
on the other hand,
the diameter of 
CL graphs is $\Theta(\log n)$~\cite{chung2002average}
while for PA graphs it is
$\Theta(\log \log n)$~\cite{diamPA}.
\tabref{diameters} provides an overview over existing results.
It was open so far how the diameter of hyperbolic random graphs compares to the aforementioned
bounds for other scale-free graph models.
The only known results for their diameter are
$\Oh((\log n)^{\frac{32}{(3-\beta)(5-\beta)}+1})$ by \citet{KiwiMitsche15}, and a polylogarithm with no explicit constant by \citet{bringmann2015geometric}.

\textbf{Our contribution.}
We improve upon the previous results as described by the following theorems.
First, we present a much simpler proof which also shows polylogarithmic upper bound for the diameter,
but with a better (i.e.\ smaller) exponent.\footnote{The conference version of this paper \cite{DBLP:conf/icalp/FriedrichK15} also contained an incorrect proof of a logarithmic upper bound on the diameter for small average degrees. In particular, Lemma~14 contained a mistake where the expected value was taken over probabilities $p_i$ that did not add up to $1$.
It is an open problem to close the gap between the polylogarithmic upper and logarithmic lower bound. }
\begin{theorem}
\label{thm:largesection}
Let $2 < \beta < 3$. The diameter of the giant component in the hyperbolic random graph $\mathcal G(n, \alpha, C)$ is $\Oh((\log n)^{\frac2 {3-\beta}})$ with probability $1 - \Oh(n^{-3/2})$.
\end{theorem}
The proof of \thmref{largesection} is presented in \secref{polyupper}.
For a lower bound on the diameter, we prove the following theorem.
\begin{theorem}
\label{thm:lowerbound}
Let $2 < \beta < 3$. Then, the diameter of the giant component in the hyperbolic random graph $\mathcal G(n, \alpha, C)$ is $\Omega(\log n)$ with probability $1 - n^{-\Omega(1)}$.
\end{theorem}
We point out that although we prove all diameter bounds on the giant component, our proofs will make apparent that the giant component is in fact the component with the largest diameter in the graph.
\medskip

\section{Notation and Preliminaries}
In this section, we briefly introduce hyperbolic random graphs. Although this paper is self-contained, we recommend to a reader who is unfamiliar with the notion of hyperbolic random graphs the more thorough investigations \cite{krioukov2010hyperbolic,gugelmann2012random}. 

Let $\mathbb H_2$ be the hyperbolic plane. Following \cite{krioukov2010hyperbolic}, we use the {\em native} representation; in which a point $v \in \mathbb H_2$ is represented by polar coordinates $(r_v, \varphi_v)$; and $r_v$ is the hyperbolic distance of $v$ to the origin.\footnote{Note that this seemingly trivial fact does not hold for conventional models (e.g.\ Poincar\'e halfplane) for the hyperbolic plane.}

To construct a hyperbolic random graph $G(n, \alpha, C)$, consider now a circle $D_n$ with radius $R = 2 \ln n + C$ that is centered at the origin of $\mathbb H_2$. Inside $D_n$, $n$ points are distributed independently as follows. For each point $v$, draw $\varphi_v$ uniformly at random from $[0,2\pi)$, and draw $r_v$ according to the probability density function
$$\rho(r) := \frac{\alpha \sinh(\alpha r)}{\cosh(\alpha R) - 1} \approx \alpha e^{\alpha(r-R)}.$$
Next, connect two points $u,v$ if their hyperbolic distance is at most $R$, i.e.\ if
\begin{equation}
\dist(u,v) := \cosh^{-1}(\cosh(r_u)\cosh(r_v) - \sinh(r_u)\sinh(r_v)\cos(\Delta\varphi_{u,v})) \leq R. \label{eq:distance}
\end{equation}
By $\Delta\varphi_{u,v}$ we describe the small relative angle between two nodes $u,v$, i.e. $\Delta\varphi_{u,v} := \cos^{-1}(\cos(\varphi_u - \varphi_v)) \leq \pi$. 

This results in a graph whose degree distribution follows a power law with exponent $\beta = 2\alpha+1$, if $\alpha \geq \tfrac12$, and $\beta = 2$ otherwise \cite{gugelmann2012random}. Since most real-world networks have been shown to have a power law exponent $2 < \beta < 3$, we assume throughout the paper that $\tfrac 12 < \alpha < 1$. \citet{gugelmann2012random} proved that the average degree in this model is then $\delta = (1+o(1))\,\frac{2 \alpha^2 e^{-C/2}}{\pi (\alpha-1/2)^2}$.

We now present a handful of Lemmas useful for analyzing the hyperbolic random graph. Most of them are taken from \cite{gugelmann2012random}. We begin by an upper bound for the angular distance between two connected nodes. Consider two nodes with radial coordinates $r,y$. Denote by $\theta_r(y)$ the maximal radial distance such that these two nodes are connected. By \eq{distance}, 
\begin{equation}
\theta_r(y) = \arccos\left(\frac{\cosh(y) \cosh(r) - \cosh(R)}{\sinh(y)\sinh(r)} \right). \label{eq:maxangle}
\end{equation}
This terse expression is closely approximated by the following Lemma.
\begin{lemma}[\cite{gugelmann2012random}]
\label{lem:maxangle}
Let $0\leq r \leq R$ and $y \geq R -r$. Then,
$$\theta_r(y) = \theta_y(r) = 2e^{\frac{R-r-y}2} (1 \pm \Theta(e^{R-r-y})). $$
\end{lemma} 
For most computations on hyperbolic random graphs, we need expressions for the probability that a sampled point falls into a certain area. To this end, \citet{gugelmann2012random} define the {\em probability measure} of a set $S \subseteq D_n$ as
$$\mu(S) := \int_S f(y) \dif y,$$
where $f(r)$ is the probability mass of a point $p=(r,\varphi)$ given by $f(r) := \tfrac{\rho(r)}{2\pi} = \frac{\alpha \sinh (\alpha r)}{2\pi(\cosh(\alpha R) -1)}$.
We further define the {\em ball} with radius $x$ around a point $(r,\varphi)$ as
$$B_{r,\varphi}(x) := \{ (r', \varphi') \mid \dist((r',\varphi'),(r,\varphi)) \leq x\}.$$
We write $B_r(x)$ for $B_{r,0}(x)$. Note that $D_n = B_0(R)$. Using these definitions, we can formulate the following Lemma.
\begin{lemma}[\cite{gugelmann2012random,KiwiMitsche15}]
\label{lem:intersection}
For any $0 \leq r \leq R$ we have
\begin{align}
&\mu(B_0(r)) = e^{-\alpha(R-r)}(1+o(1)) \label{eq:smallball} \\
&\mu(B_r(R) \cap B_0(R-m)) = \tfrac{2\alpha}{\pi(\alpha - 1/2)} \cdot e^{-\alpha m - \frac12(r - m)} + \Oh(e^{-\alpha r}) \label{eq:lens}
\end{align}
\end{lemma}
Since we often argue over sequences of nodes on a path, we say that a node~$v$ is {\em between} two nodes $u,w$, if $ \Delta\varphi_{u,v} + \Delta\varphi_{v,w} = \Delta\varphi_{u,w}.$ 
Recall that $\Delta\varphi_{u,v} \leq \pi$ describes the small angle between $u$ and $v$. E.g., if $u = (r_1, 0), v = (r_2, \tfrac\pi2), w = (r_3, \pi)$, then $v$ lies between $u$ and $w$. However, $w$ does not lie between $u$ and $v$ as $\Delta\varphi_{u,v} = \pi/2$ but $\Delta\varphi_{u,w} + \Delta\varphi_{w,v} = \tfrac34 \pi$.

Finally, we define the area $B_I := B_0(R-\tfrac{\log R}{1-\alpha} - c)$ as the {\em inner band}, and $B_O := D_n \setminus B_I$ as the {\em outer band}, where $c \in \mathbb R$ is a large enough constant. 

\medskip \noindent \textbf{The Poisson Point Process.}
We often want to argue about the probability that an area $S \subseteq D_n$ contains one or more nodes. To this end, we usually apply the simple formula
\begin{equation} \Pr[\exists v \in S] = 1 - (1 - \mu(S))^n \geq 1 - \exp(-n \cdot \mu(S)). \label{eq:existsnode} \end{equation}
Unfortunately, this formula significantly complicates once the positions of some nodes are already known. This introduces conditions on $\Pr[\exists v \in S]$ which can be hard to grasp analytically. To circumvent this problem, we use a Poisson point process $\mathcal P_n$ \cite{penrose2003random} which describes a different way of distributing nodes inside $D_n$. It is fully characterized by the following two properties:
\begin{itemize}
\item If two areas $S,S'$ are disjoint, then the number of nodes that fall within $S$ and $S'$ are independent random variables.
\item The expected number of points that fall within $S$ is $\int_S n \mu(S)$.
\end{itemize}
One can show that these properties imply that the number of nodes inside $S$ follows a Poisson distribution with mean $n\mu(S)$. In particular, we obtain that the number of nodes $|\mathcal P_n|$ inside $D_n$ is distributed as $\mathrm{Po}(n)$, i.e.\ $\Ex[|P_n|] = n$, and 
$$\Pr(|\mathcal P_n| = n) = \frac{e^{-n} n^n}{n!} = \Theta(n^{-1/2}).$$
Let the random variable $\mathcal G(\mathcal P_n, n, \alpha, C)$ denote the resulting graph when using the Poisson point process to distribute nodes inside $D_n$. Since it holds
$$\Pr[\mathcal G(\mathcal P_n, n, \alpha, C) = G \mid |\mathcal P_n| = n] = \Pr[\mathcal G(n,\alpha,C) = G],$$
we have that every property $p$ with $\Pr[p(\mathcal G(\mathcal P_n, n, \alpha, C))] \leq \Oh(n^{-c})$ holds for the hyperbolic random graphs with probability $\Pr[p(\mathcal G(n,\alpha,C))] \leq \Oh(n^{\frac12 -c})$.

We explicitly state whenever we use the Poisson point process $\mathcal G(\mathcal P_n,n,\alpha,C)$ instead of the normal hyperbolic random graph $\mathcal G(n,\alpha,C)$. In particular, we can use a matching expression for \eq{existsnode}:
$\Pr[\exists v \in S] = 1 - \exp(-n \cdot \mu(S)).$
\section{Polylogarithmic Upper Bound}
\label{sec:polyupper}
As an introduction to the main proof, we first show a simple polylogarithmic upper bound on the diameter of the hyperbolic random graph. We start by  investigating nodes in the inner band $B_I$ and show that they are connected by a path of at most $\Oh(\log \log n)$ nodes.  
%
We prove this by partitioning $D_n$ into $R$ layers of constant thickness $1$. Then, a node in layer $i$ has radial coordinate $\in (R - i, R-i+1]$. We denote the layer $i$ by $L_i := B_0(R-i+1) \setminus B_0(R-i)$.
\begin{lemma}
\label{lem:constbounds}
Let $1 \leq i,j \leq R/2$, and consider two nodes $v \in L_i, w \in L_j$. Then,
$$\frac2ee^{\frac{i+j-R}{2}} (1 - \Theta(e^{i+j-R})) \leq \theta_{r_u}(r_v) \leq  2e^{\frac{i+j-R}{2}} (1 + \Theta(e^{i+j-R})),$$
Furthermore, we have
$\mu(L_j \cap B_R(v)) = \Theta(e^{-\alpha j + \frac{i+j-R}{2}}),$
and, if $(i+j)/R < 1-\eps$ for some constant $\eps > 0$, we have for large $n$
$$\frac1e e^{-\alpha j + \frac{i+j-R}{2}} \leq \mu(L_j \cap B_R(v)) \leq 4 e^{-\alpha j + \frac{i+j-R}{2}}. $$
\end{lemma}
\begin{proof}
The statements follow directly from \lemrefs{maxangle}{intersection} and the fact that we have $R-i < r_v \leq R - i +1$ for a node $v \in L_i$.
\end{proof}
Using \lemref{constbounds}, we can now prove that a node $v \in B_I$ has a path of length $\Oh(\log \log n)$ that leads to $B_0(R/2)$. Recall that the inner band was defined as $B_I := B_0(R-\tfrac{\log R}{1-\alpha} - c)$, where $c$ is a large enough constant.
\begin{lemma}
\label{lem:diamlarge}
Consider a node $v$ in layer $i$. With probability $1-\Oh(n^{-3})$ it holds 
\begin{enumerate}
\item if $i \in [\tfrac{\log R}{1-\alpha} + c, \tfrac{2\log R}{1-\alpha} +c]$, then $v$ has a neighbor in layer $L_{i+1}$, and
\item if $i \in [\tfrac{2\log R}{1-\alpha} +c, R/2] $, then $v$ has a neighbor in layer $L_j$ for $j=\tfrac{\alpha}{2\alpha-1} i$.
\end{enumerate}
\end{lemma}
\begin{proof}
The probability that node $v \in L_i$ does not contain a neighbor in $L_{i+1}$ is
\begin{align*}
(1-\Theta(e^{-\alpha (i+1) + i + \frac{1-R}{2}}))^n 
\leq \exp(-\Theta(1) \cdot e^{\log R + c(1-\alpha)}).
\end{align*}
Since $R = 2\log n +C$ and $c$ is a large enough constant, this proves part (1) of the claim. An analogous argument shows part (2).
\end{proof}
\lemref{diamlarge} shows that there exists a path of length $\Oh(\log \log n)$ from each node $v\in B_I$  to some node $u \in B_0(R - \tfrac{2\log R}{1-\alpha} -c)$. Similarly, from $u$ there exists a path of length $\Oh(\log \log n)$ to $B_0(R/2)$ with high probability. Since we know that the nodes in $B_0(R/2)$ form a clique by the triangle inequality, we therefore obtain that all nodes in $B_I$ form a connected component with diameter $\Oh(\log \log n)$.
\begin{corollary}
\label{cor:smalldiam}
Let $\tfrac12 < \alpha < 1$. With probability $1-\Oh(n^{-3})$, all nodes $u,v \in B_I$ in the hyperbolic random graph are connected by a path of length $\Oh(\log \log n)$.
\end{corollary} 

\subsection{Outer Band}
\label{sec:outer}

By \corref{smalldiam}, we obtain that the diameter of the graph induced by nodes in $B_I$ is at most $\Oh(\log\log n)$. In this section, we show that each component in $B_O$ has a polylogarithmic diameter. Then, one can easily conclude that the overall diameter of the giant component is polylogarithmic, since all nodes in $B_0(R/2)$ belong to the giant component \cite{bfmconnected}.
We begin by presenting one of the crucial Lemmas in this paper that will often be reused. 
\begin{lemma}
\label{lem:underpass}
Let $u,v,w \in V$ be nodes such that $v$ lies between $u$ and $w$, and let $\{u,w\} \in E$. If $r_v \leq r_u$ and $r_v \leq r_w$, then $v$ is connected to both $u$ and $w$. If $r_v \leq r_u$ but $r_v \geq r_w$, then $v$ is at least connected to $w$. 
\end{lemma}
\begin{proof}
By \cite[Lemma 5.28]{bfmconnected}, we know that if two nodes $(r_1, \varphi_1), (r_2,\varphi_2)$ are connected, then so are $(r_1', \varphi_1), (r_2', \varphi_2)$ where $r_1 \leq r_1'$ and $r_2' \leq r_2$. Since the distance between nodes is monotone in the relative angle $\Delta\varphi$, this proves the first part of the claim. The second part can be proven by an analogous argument.
\end{proof}
For convenience, we say that an edge $\{u,w\}$ {\em passes under} $v$ if one of the requirements of \lemref{underpass} is fulfilled. Using this, we are ready to show \thmref{largesection}. In this argument, we investigate the angular distance a path can at most traverse until it passes under a node in $B_I$. By \lemref{underpass}, we then have with high probability a short path to the center $B_0(R/2)$ of the graph.

\begin{proof}[(Proof of \thmref{largesection})]
Partition the hyperbolic disc into $n$ disjoint sectors of equal angle $\Theta(1/n)$. The probability that $k$ consecutive sectors contain no node in $B_I$ is
\begin{align*}
(1- \Theta(k/n) \cdot \mu(B_0(R-\tfrac{\log R}{1-\alpha} - c)))^n &\leq \exp(- \Theta(1) \cdot k \cdot e^{-\alpha\log R/(1-\alpha)}) \\
&= \exp(-\Theta(1) \cdot k \cdot (\log n)^{-\frac\alpha{1-\alpha}}).
\end{align*}
Hence, we know that with probability $1- \Oh(n^{-3})$, there are no $k := \Theta((\log n)^{\frac{1}{1-\alpha}})$ such consecutive sectors. By a Chernoff bound, the number of nodes in $k$ such consecutive sectors is $\Theta((\log n)^{\frac1{1-\alpha}})$ with probability $1-\Oh(n^{-3})$. Applying a union bound, we get that with probability $1- \Oh(n^{-2})$, every sequence of $k$ consecutive sectors contains at least one node in $B_I$ and at most $\Theta(k)$ nodes in total. Consider now a node $v\in B_O$ that belongs to the giant component. Then, there must exist a path from $v$ to some node $u \in B_I$. By \lemref{underpass}, this path can visit at most $k$ sectors---and therefore use at most $\Theta(k)$ nodes---before reaching~$u$. From $u$, there is a path of length $\Oh(\log\log n)$ to the center $B_0(R/2)$ of the hyperbolic disc by \corref{smalldiam}. Since this holds for all nodes, and the center forms a clique, the diameter is therefore $\Oh((\log n)^{\frac{1}{1-\alpha}}) = \Oh((\log n)^{\frac{2}{3-\beta}})$.
\end{proof}
From the proof it follows that every component inhabiting $\Omega((\log n)^{\frac{1}{1-\alpha}})$ sectors is connected to the center. We derive the following Corollary.
\begin{corollary}
\label{cor:2ndcomp}
Let $2 < \beta < 3$. The second largest component of the hyperbolic random graph is of size at most $\Oh((\log n)^{\frac{2}{3-\beta}})$ with probability $1 - \Oh(n^{-3/2})$.
\end{corollary}
These bounds improves upon the results in \cite{KiwiMitsche15} who show an upper bound of $\Oh((\log n)^{\frac{32}{(3-\beta)(5-\beta)}+1})$ on the diameter and $\Oh((\log n)^{\frac{64}{(3-\beta)(5-\beta)}+1})$ on the second largest component.
As we will see in \thmref{lowerbound}, however, the lower bound on the diameter is only $\Omega(\log n)$. It is an open problem to bridge this gap.

\section{Logarithmic Lower Bound}
\label{sec:lower}

\citet{KiwiMitsche15} provide a proof for the existence of a path component of length $\Theta(\log n)$ with high probability. In this section, we show a stronger statement, namely that the {\em largest} component has a diameter of $\Omega(\log n)$. This proves the intuition that the component with the largest diameter is in fact the giant component, which is not obvious a priori. 

\begin{proof}[Proof of \thmref{lowerbound}]
Let $\eps := (\frac12-\frac1{4\alpha})$. Observe that for $\tfrac12 < \alpha <1$, we have $0 < \eps < \tfrac14$. Consider the model $\mathcal G(n,\alpha, C)$, i.e.\ {\em not} the Poisson point process. With high probability, there are no nodes in $B_0(\eps R)$:
\begin{align*}
\Pr[B_0(\eps R) = \emptyset]  &= \exp(-\Theta(1) \cdot e^{R/2} \cdot e^{-(\frac \alpha2 + \frac14)R)} ) \\
&= 1 - \Theta(1) \cdot e^{(\frac14 - \frac\alpha2)R}.\\
&= 1 - n^{-\Omega(1)}.
\end{align*}
In the following, we condition on the fact that there are no such nodes; and switch to the Poisson point process. Consider now a node $v \in L_1$. The largest angular distance $v$ can have to one of its neighbors is
\begin{equation}
\Delta\varphi \leq 2e^{-\frac{\eps R}2} (1 \pm \Oh(e^{-\eps R})) \leq \Oh(n^{-\eps}). \label{eq:uncoveredangle}
\end{equation}

Similarly to \thmref{largesection}, we partition the Disc $D_n$ into $\Theta(n)$ sectors of equal angle $\varphi := e^{-R/2} = \Theta(n)$. Then, two nodes $u,v\in L_1$ in neighboring sectors have angular distance at most  $2e^{-R/2} $, and are therefore connected. On the flip side, two nodes with at least $6$ sectors between them have no edge, since their angle is $6e^{-R/2} > 2e^{-R/2 +1} (1 + \Oh(e^{-R}))$.

Consider now $p$ consecutive sectors, where $p$ is to be fixed later. For each of these $p$ sectors, the probability that it contains exactly one node in $L_1$ is 
$\geq (e^{-R/2} \cdot n e^{-\alpha} ) = e^{-\Theta(1)}$, i.e.\ a constant smaller than $1$. The probability that this node has no further neighbors (apart from the neighbors in $L_1$ in the other sectors) is again $e^{-\Theta(1)}$ by \lemref{intersection}. We name these nodes $v_1, \ldots, v_p$.

Similarly, the probability that sector $p+1$ contains exactly one node $v_{p+1}$ in $L_3$ is again $e^{-\Theta(1)}$. From here, we expose a path to the inner band $B_I$ as follows. Assume we have a node $v \in L_i$. Assume further that all nodes $v_1, \ldots, v_{p}$ are to the left of $v$. Then, we consider the probability that $v$ has a neighbor in layer $L_j$ for $j = \tfrac 1{2\alpha-1}i$, while we condition on the fact that none of the nodes $v_1, \ldots, v_p$ have neighbors in the upper layers as stated before. By \lemref{constbounds} this means that in layer $L_j$, we have not yet uncovered an angle of at least 
\begin{align*}
&\tfrac2e e^{(i+\frac i {2\alpha-1} -R)/2} - 3e^{(\frac i {2\alpha-1}-R)/2} 
\geq{} \Theta(1) \cdot e^{(i+\frac i {2\alpha-1} -R)/2},
\end{align*}
as $\tfrac {3e}2 e^{-3/2} < 1$. Therefore, the probability that node $v$ has a neighbor in layer $L_j$ that is not connected to $v_1, \ldots, v_p$, is at least
$$1 - \exp(- \Theta(1) \cdot n \cdot e^{-\frac{\alpha i}{2\alpha-1} } e^{(i+\frac i {2\alpha-1} -R)/2}) =  e^{-\Theta(1)} < 1.$$
In total, the probability that $v_1, \ldots, v_p$ exist as described above; and that they are connected to $B_I$ is thereby
$e^{-\Theta(p + \log\log\log n)}. $

Furthermore, by \eq{uncoveredangle}, we know that when exposing this information we at most expose an angle of $\Oh(\tfrac pn + n^{-\Omega(1)} + \log\log\log n \cdot (\log n)^{1/(1-\alpha)} )$ of the graph. 
Therefore, if $\tfrac pn < n^{-\Omega(1)}$, we can repeat this experiment independently $n^{\Omega(1)}$ times. The probability that all of them fail is at most
$$(1 - e^{-\Theta(p+\log\log\log n)})^{n^{\Omega(1)}} = \exp(- e^{-\Theta(p)} n^{\Omega(1)}) = \exp(- n^{\Omega(1)}),$$  
if $p=\Theta(\log n)$ is chosen small enough. This proves the claim.
\end{proof}

\section{Conclusion}
We derive a new polylogarithmic upper bound on the diameter of hyperbolic random graphs; and 
further prove a logarithmic lower bound. This  immediately yields lower bounds for any broadcasting protocol that has to reach all nodes. 
Processes such as bootstrap percolation or rumor spreading therefore must run at least $\Omega(\log n)$ steps until they inform all nodes in the giant component. In particular, this result stands in contrast to the average distance of two nodes in the hyperbolic random graph, which is of order $\Theta(\log \log n)$ \cite{bringmann2015geometric,abdullah2015typical}. This implies the existence of a path that is exponentially longer than the average path.

Our work focuses on power law exponents $2<\beta<3$, but we believe that it is possible to bound the diameter for $\beta>3$ by $\Theta(\log n)$. For other scale-free models it was also interesting to study the phase transition at $\beta=2$ and $\beta=3$.

\bibliography{bibliography.bib}

\end{document}

%% file: makros_journal.tex
\usepackage{etex}
\reserveinserts{28}

\usepackage[american]{babel}
\usepackage [autostyle]{csquotes}
\MakeOuterQuote{"}

\usepackage{amssymb, amsmath}
\usepackage{graphicx}

\usepackage{amsthm}

\usepackage{mathtools,array}
\usepackage{amsfonts}
\usepackage{amssymb,bm,bbm}

\PassOptionsToPackage{hyphens}{url}
\usepackage[linktocpage,
            colorlinks=true,linkcolor=black,
            citecolor=black,urlcolor=black,
            pdfpagemode=UseNone,
            bookmarks=true,
            breaklinks,
            bookmarksopen=false,
            hyperfootnotes=false,
            pdfhighlight=/I,
            pdftitle={On the diameter of hyperbolic random graphs},
            pdfauthor={Friedrich, Krohmer}
            ]{hyperref}
            
\usepackage{thmtools, thm-restate}
\declaretheorem[name=Theorem]{theorem}
\declaretheorem[sibling=theorem,name=Lemma]{lemma}
\declaretheorem[sibling=theorem,name=Definition]{definition}
\declaretheorem[sibling=theorem,name=Corollary]{corollary}

\usepackage{xspace}
\usepackage{etex}  
\usepackage{url}
\usepackage{tikz}
\usetikzlibrary{decorations.pathmorphing} 
\usetikzlibrary{decorations.pathreplacing,calc}
\usetikzlibrary{graphs}
\usetikzlibrary{arrows}
 \usepackage{subfig}
 \usepackage{floatrow}

\usepackage{float}
\floatstyle{ruled}
\newfloat{myalgorithm}{thp}{lop}
\floatname{myalgorithm}{Algorithm}

\usepackage{enumitem}
\usepackage{hyperref}
\usepackage{algorithm,booktabs}
\usepackage[noend]{algpseudocode}

\usepackage[numbers,longnamesfirst,sort&compress,sectionbib]{natbib}
 \usepackage[labelfont=bf,font=footnotesize,indention=0.3cm,margin=0cm]{caption}  

\usepackage{setspace,color}   
\newcounter{nummer}

\newcommand{\Oh}{\mathcal{O}}

\newcommand{\dif}{\,\mathrm{d}}
\newcommand{\dist}{\mathrm{d}}

\newcommand{\IGNORE}[1]{}

\def\Ex{\mathbb{E}}
\def\Pr{\operatorname{Pr}}

\newcommand{\tabref}[1]{Table~\ref{tab:#1}}

\newcommand{\thmref}[1]{Theorem~\ref{thm:#1}}

\newcommand{\lemref}[1]{Lemma~\ref{lem:#1}}
\newcommand{\lemrefs}[2]{Lemmas~\ref{lem:#1} and~\ref{lem:#2}}
\newcommand{\corref}[1]{Corollary~\ref{cor:#1}}
\newcommand{\secref}[1]{Section~\ref{sec:#1}}

\newcommand{\eq}[1]{equation~\eqref{eq:#1}}

\renewcommand{\epsilon}{\ensuremath{\varepsilon}}
\newcommand{\eps}{\ensuremath{\varepsilon}}

\newcommand{\temporary}[1]{}

\let\oldsqrt\sqrt
\def\hksqrt{\mathpalette\DHLhksqrt}
\def\DHLhksqrt#1#2{\setbox0=\hbox{$#1\oldsqrt{#2\,}$}\dimen0=\ht0
   \advance\dimen0-0.2\ht0
   \setbox2=\hbox{\vrule height\ht0 depth -\dimen0}%
   {\box0\lower0.4pt\box2}}
\renewcommand{\sqrt}{\hksqrt}
\renewcommand{\leq}{\leqslant}
\renewcommand{\geq}{\geqslant}

\catcode`@=11
\def\nphantom{\v@true\h@true\nph@nt}
\def\nvphantom{\v@true\h@false\nph@nt}
\def\nhphantom{\v@false\h@true\nph@nt}
\def\nph@nt{\ifmmode\def\next{\mathpalette\nmathph@nt}%
  \else\let\next\nmakeph@nt\fi\next}
\def\nmakeph@nt#1{\setbox\z@\hbox{#1}\nfinph@nt}
\def\nmathph@nt#1#2{\setbox\z@\hbox{$\m@th#1{#2}$}\nfinph@nt}
\def\nfinph@nt{\setbox\tw@\null
  \ifv@ \ht\tw@\ht\z@ \dp\tw@\dp\z@\fi
  \ifh@ \wd\tw@-\wd\z@\fi \box\tw@}
\newcount\minute \newcount\hour \newcount\hourMins
\def\now{\minute=\time \hour=\time \divide \hour by 60 \hourMins=\hour \multiply\hourMins by 60
  \advance\minute by -\hourMins \zeroPadTwo{\the\hour}:\zeroPadTwo{\the\minute}}

\def\today{\the\year-\zeroPadTwo{\the\month}-\zeroPadTwo{\the\day}}
\def\zeroPadTwo#1{\ifnum #1<10 0\fi #1}

%% file: journal.bbl
\begin{thebibliography}{25}
\providecommand{\natexlab}[1]{#1}
\providecommand{\url}[1]{\texttt{#1}}
\expandafter\ifx\csname urlstyle\endcsname\relax
  \providecommand{\doi}[1]{doi: #1}\else
  \providecommand{\doi}{doi: \begingroup \urlstyle{rm}\Url}\fi

\bibitem[Abdullah et~al.(2015)Abdullah, Bode, and
  Fountoulakis]{abdullah2015typical}
M.~A. Abdullah, M.~Bode, and N.~Fountoulakis.
\newblock Typical distances in a geometric model for complex networks.
\newblock \emph{arXiv preprint arXiv:1506.07811}, 2015.

\bibitem[Barab{\'a}si and Albert(1999)]{barabasi1999emergence}
A.-L. Barab{\'a}si and R.~Albert.
\newblock Emergence of scaling in random networks.
\newblock \emph{Science}, 286:\penalty0 509--512, 1999.

\bibitem[Bode et~al.(2013)Bode, Fountoulakis, and M{\"u}ller]{Bode:2013aa}
M.~Bode, N.~Fountoulakis, and T.~M{\"u}ller.
\newblock On the giant component of random hyperbolic graphs.
\newblock In \emph{7th European Conference on Combinatorics, Graph Theory and
  Applications (EuroComb)}, pp. 425--429. 2013.

\bibitem[Bode et~al.(2014)Bode, Fountoulakis, and M{\"u}ller]{bfmconnected}
M.~Bode, N.~Fountoulakis, and T.~M{\"u}ller.
\newblock The probability that the hyperbolic random graph is connected.
\newblock \url{www.math.uu.nl/~Muell001/Papers/BFM.pdf}, 2014.

\bibitem[Bogun{\'a} et~al.(2010)Bogun{\'a}, Papadopoulos, and
  Krioukov]{boguna2010sustaining}
M.~Bogun{\'a}, F.~Papadopoulos, and D.~Krioukov.
\newblock Sustaining the internet with hyperbolic mapping.
\newblock \emph{Nature Communications}, 1:\penalty0 62, 2010.

\bibitem[Bollob{\'a}s(1998)]{bollobas1998random}
B.~Bollob{\'a}s.
\newblock \emph{Random graphs}.
\newblock Springer, 1998.

\bibitem[Bollob{\'{a}}s and Chung(1988)]{BollobasC88}
B.~Bollob{\'{a}}s and F.~R.~K. Chung.
\newblock The diameter of a cycle plus a random matching.
\newblock \emph{{SIAM} Journal of Discrete Mathematics}, 1:\penalty0 328--333,
  1988.

\bibitem[Bringmann et~al.(2015)Bringmann, Keusch, and
  Lengler]{bringmann2015geometric}
K.~Bringmann, R.~Keusch, and J.~Lengler.
\newblock Geometric inhomogeneous random graphs.
\newblock \emph{arXiv preprint arXiv:1511.00576}, 2015.

\bibitem[Candellero and Fountoulakis(2014)]{candellero2014bootstrap}
E.~Candellero and N.~Fountoulakis.
\newblock Bootstrap percolation and the geometry of complex networks, 2014.
\newblock \arxiv{1412.1301}.

\bibitem[Chung and Lu(2002)]{chung2002average}
F.~Chung and L.~Lu.
\newblock The average distances in random graphs with given expected degrees.
\newblock \emph{Proceedings of the National Academy of Sciences}, 99:\penalty0
  15879--15882, 2002.

\bibitem[Dommers et~al.(2010)Dommers, van~der Hofstad, and
  Hooghiemstra]{diamPA}
S.~Dommers, R.~van~der Hofstad, and G.~Hooghiemstra.
\newblock Diameters in preferential attachment models.
\newblock \emph{Journal of Statistical Physics}, 139:\penalty0 72--107, 2010.

\bibitem[Fountoulakis et~al.(2015)Fountoulakis, Panagiotou, and
  Sauerwald]{UltraFastRumor:unpub}
N.~Fountoulakis, K.~Panagiotou, and T.~Sauerwald.
\newblock Ultra-fast rumor spreading in models of real-world networks, 2015.
\newblock Unpublished draft.

\bibitem[Friedrich and
  Krohmer(2015{\natexlab{a}})]{DBLP:conf/icalp/FriedrichK15}
T.~Friedrich and A.~Krohmer.
\newblock On the diameter of hyperbolic random graphs.
\newblock In \emph{42nd International Colloquium on Automata, Languages, and
  Programming, ({ICALP})}, pp. 614--625, 2015{\natexlab{a}}.

\bibitem[Friedrich and Krohmer(2015{\natexlab{b}})]{friedrich2015cliques}
T.~Friedrich and A.~Krohmer.
\newblock Cliques in hyperbolic random graphs.
\newblock In \emph{34th IEEE Conference on Computer Communications (INFOCOM)},
  2015{\natexlab{b}}.
\newblock \href{http://docs.theinf.uni-jena.de/paper/2015INFOCOM.pdf}{To
  appear}.

\bibitem[Friedrich et~al.(2013)Friedrich, Sauerwald, and Stauffer]{diamEucRGG}
T.~Friedrich, T.~Sauerwald, and A.~Stauffer.
\newblock Diameter and broadcast time of random geometric graphs in arbitrary
  dimensions.
\newblock \emph{Algorithmica}, 67:\penalty0 65--88, 2013.

\bibitem[Gugelmann et~al.(2012)Gugelmann, Panagiotou, and
  Peter]{gugelmann2012random}
L.~Gugelmann, K.~Panagiotou, and U.~Peter.
\newblock Random hyperbolic graphs: degree sequence and clustering.
\newblock In \emph{39th International Colloquium on Automata, Languages, and
  Programming (ICALP)}, pp. 573--585, 2012.

\bibitem[Kiwi and Mitsche(2015)]{KiwiMitsche15}
M.~Kiwi and D.~Mitsche.
\newblock A bound for the diameter of random hyperbolic graphs.
\newblock In \emph{12th Workshop on Analytic Algorithmics and Combinatorics
  (ANALCO)}, pp. 26--39, 2015.

\bibitem[Kleinberg(2000)]{Kleinberg00}
J.~Kleinberg.
\newblock Navigation in a small world.
\newblock \emph{Nature}, 406:\penalty0 845, 2000.

\bibitem[Krioukov et~al.(2010)Krioukov, Papadopoulos, Kitsak, Vahdat, and
  Bogu{\~n}{\'a}]{krioukov2010hyperbolic}
D.~Krioukov, F.~Papadopoulos, M.~Kitsak, A.~Vahdat, and M.~Bogu{\~n}{\'a}.
\newblock Hyperbolic geometry of complex networks.
\newblock \emph{Physical Review E}, 82:\penalty0 036106, 2010.

\bibitem[Martel and Nguyen(2004)]{MartelN04}
C.~U. Martel and V.~Nguyen.
\newblock Analyzing {Kleinberg's} (and other) small-world models.
\newblock In \emph{23rd Annual {ACM} Symposium on Principles of Distributed
  Computing (PODC)}, pp. 179--188, 2004.

\bibitem[Papadopoulos et~al.(2014)Papadopoulos, Psomas, and Krioukov]{6705650}
F.~Papadopoulos, C.~Psomas, and D.~Krioukov.
\newblock Network mapping by replaying hyperbolic growth.
\newblock \emph{IEEE/ACM Transactions on Networking}, pp. 198--211, 2014.

\bibitem[Penrose(2003)]{penrose2003random}
M.~Penrose.
\newblock \emph{Random Geometric Graphs}.
\newblock Oxford scholarship online. Oxford University Press, 2003.

\bibitem[Riordan and Wormald(2010)]{riordan2010diameter}
O.~Riordan and N.~Wormald.
\newblock The diameter of sparse random graphs.
\newblock \emph{Combinatorics, Probability and Computing}, 19:\penalty0
  835--926, 2010.

\bibitem[von Looz et~al.(2015)von Looz, Staudt, Meyerhenke, and
  Prutkin]{von2015fast}
M.~von Looz, C.~L. Staudt, H.~Meyerhenke, and R.~Prutkin.
\newblock Fast generation of dynamic complex networks with underlying
  hyperbolic geometry, 2015.
\newblock \arxiv{1501.03545}.

\bibitem[Watts and Strogatz(1998)]{WattsStrogatz98}
D.~J. Watts and S.~H. Strogatz.
\newblock Collective dynamics of `small-world' networks.
\newblock \emph{Nature}, 393:\penalty0 440--442, 1998.

\end{thebibliography}
